\def\reals{{I\kern-.35em R}}

\documentclass{article}
\usepackage[utf8]{inputenc}
\usepackage{lscape}
\usepackage{indentfirst}
\usepackage[margin=1in]{geometry}
\usepackage[margin=0.5in]{caption}
\usepackage[numbers]{natbib}
\PassOptionsToPackage{hyphens}{url}\usepackage{hyperref}
\usepackage{graphicx}
\usepackage{array,booktabs}
\newcolumntype{C}{>{$\displaystyle}c<{$}}
\usepackage{verbatim}
\usepackage{caption}
\usepackage{subcaption}
\usepackage{amsthm}
\usepackage{amsmath}
\usepackage{amsfonts}
\usepackage{bbm}
\usepackage[english]{babel}
\usepackage{amssymb}
\usepackage{amsthm}
\usepackage{authblk}
\newtheorem{theorem}{Theorem}
\usepackage{hyperref}
\usepackage{mathtools} 
\usepackage{xurl}
\usepackage{silence}
\newtheorem{lemma}[theorem]{Lemma}
\title{\LARGE f-Betas and Portfolio Optimization with f-Divergence induced Risk Measures}

\author[$\star$]{Rui Ding}

\affil[$\star$]{\small Department of Applied Mathematics and Statistics, Stony Brook University, Stony Brook, NY, USA, 11794}
\affil[$\star$]{Correspondence: rui.ding.1@stonybrook.edu}

\date{}

\begin{document}
\maketitle
\abstract{In this paper, we build on using the class of f-divergence induced coherent risk measures for portfolio optimization and derive its necessary optimality conditions formulated in CAPM format. We derive a new f-Beta similar to the Standard Betas and also extended it to previous works in Drawdown Betas. The f-Beta evaluates portfolio performance under an optimally perturbed market probability measure, and this family of Beta metrics gives various degrees of flexibility and interpretability. We conduct numerical experiments using selected stocks against a chosen S\&P 500 market index as the optimal portfolio to demonstrate the new perspectives provided by Hellinger-Beta as compared with Standard Beta and Drawdown Betas. In our experiments, the squared Hellinger distance is chosen to be the particular choice of the f-divergence function in the f-divergence induced risk measures and f-Betas. We calculate Hellinger-Beta metrics based on deviation measures and further extend this approach to calculate Hellinger-Betas based on drawdown measures, resulting in another new metric which is termed Hellinger-Drawdown Beta. We compare the resulting Hellinger-Beta values under various choices of the risk aversion parameter to study their sensitivity to increasing stress levels. }

\section{Introduction}
The Capital Asset Pricing Model (CAPM) (\citet{10}, \citet{11}) is a fundamental model in portfolio theory and risk management. It is based on a  Markowitz mean-variance portfolio optimization problem (\citet{4}). Tremendous literature is available on CAPM, see for instance, critical review papers by \citet{G2004}, and \citet{R2016}. The Standard Beta relates the expected return of a security and the expected excess return of a market.
Since its introduction, Beta has been used as a key indicator of asset performance in portfolio management. The variance risk measure used in the standard CAPM formulation has a conceptual drawback: it does not distinguish losses and gains of a portfolio.
For this reason, \citet{Mark1959} considered Semi-Variance 
based only on negative returns, and the associated Beta was called
Downside Beta. Although the idea sounds conceptually attractive, Downside Beta and Standard Beta have close values therefore, it provides little information in addition to Standard Beta. In the vast literature on CAPM models, various non-symmetric risk measures have been proposed as an alternative to variance. In particular, Conditional Value-at-Risk (CVaR) is a popular choice of risk measure. It was introduced by \citet{5} for continuous distributions as the conditional expected loss exceeding Value-at-Risk (VaR) and generalized to discrete distributions in \citet{6}. CAPM has been extended to non-symmetric risk measures such as Generalized Deviations by \citet{14}, which demonstrated that CAPM equations are necessary optimality conditions for portfolio optimization problems. In particular, Beta was computed for CVaR and Lower Semi-Deviation (the square root of Semi-Variance). See the review paper by \citet{15} for detailed discussions on these and other non-symmetric risk measures, which provide formulas for Betas.

In this paper, we discuss a new Beta metric based on a recently studied family of return-based risk measures called f-divergence induced risk measures. This is a coherent risk measure for every choice f-divergence function and the statistical confidence radius. \citet{DP20} contains a detailed study of the properties and characterizations of this risk measure. Based on their theoretical results, as reviewed in Section 2, we discuss portfolio optimization problems with f-divergence induced risk measures in Section 3 and propose in Section 4 a family of Beta metrics called f-Beta that relates the losses of a stock to that of the market under given optimal perturbations as characterized through a dual maximization problem for the f-divergence induced risk measure. The f-divergence induced risk measures have more desirable properties as compared with CVaR risk or VaR. Compared with CVaR, which shares very similar definitions, the f-divergence induced risk measures are smoother, and the corresponding risk identifier functions are more continuous with respect to the changing risk aversion parameter. This means that attention (weights) on the stressful scenarios being considered are transitioning more gradually. Moreover, the choice of the generating divergence function provides flexibility in characterizing the shape of the risk measure and hence characterizes different types of risk aversion behaviors. In particular, we can interpret f-divergence induced risk measures as views on the worst-case expected loss under assumptions of various degrees of market distribution shift, where the distance between distributions is measured using well-defined statistical divergences/distances. Both facts make it more natural to be considered as a risk measure. Similar to CVaR under smaller tail risk levels (with the tail risk level of 1 being the risk-neutral expectation), f-divergence induced risk measures transition from risk-neutral expectation to worst-case loss as the risk aversion parameter(which is also a statistical confidence region radius) increases from zero to larger values, hence providing a complete and continuous middle-ground between risk-neutral and worst-case loss averse. See more related discussion in Section 5. We expect the desirable properties of f-divergence induced risk measures to translate to corresponding Beta metrics based on such risk measures in the CAPM framework. Numerical results of Hellinger-Beta, based on choosing squared Hellinger distance as the f-divergence function, is demonstrated in Section 6 using DOW 30 data.

We further build on f-Betas and extend them to previous work on Drawdown Betas. A considerable drawback of Variance, CVaR, Semi-Deviation, and many other risk measures is that they are static characteristics, which do not account for persistent consecutive portfolio losses (may be resulting in a large cumulative loss). The dynamic drawdown risk measure is actively used in  portfolio management as an alternative to static measures. Portfolio managers try to build portfolios with low drawdowns. The most popular drawdown characteristic is the Maximum Drawdown. However, the Maximum Drawdown is not the best risk measure from a practical perspective: it accounts for only one specific event on a price sample path. For instance,
\citet{G2017} suggested the so-called Conditional Expected Drawdown (CED), which is the tail mean of the maximum drawdown distribution. \citet{1} proposed Conditional Drawdown-at-Risk (CDaR), which averages a specified percentage of the largest portfolio drawdowns over an investment horizon. CDaR is defined as CVaR of the drawdown observations of the portfolio cumulative returns. CDaR possesses the theoretical properties of a deviation measure, see \citet{2}. \citet{13} developed CAPM relationships based on CDaR. The paper derived necessary optimality conditions for CDaR portfolio optimization, which resulted in a CDaR Beta relating cumulative returns of a market (optimal portfolio) and individual securities. \citet{DU22} extended this previous result to consider Expected Regret of Drawdowns (ERoD) based on a reformulation of the CDaR risk through expected regret where a clear interpretation of the chosen drawdown threshold can be provided instead of using a predefined tail risk level as in CDaR Beta. Their proposed metric is called ERoD Beta, which also relates consecutive returns of a security to those of the market. A negative ERoD Beta identifies a security that has positive returns when the market has drawdowns exceeding the threshold. See also \citet{DBW}. In Section 7, we propose to use an f-divergence induced risk measure for drawdown observations to obtain a new f-Drawdown Beta and demonstrate their new perspectives via numerical results.

\section{Preliminaries on f-Divergence induced Risk Measures: Formulations and Properties}

We provide in this section a brief review of the f-divergence induced coherent risk measures and their properties, as introduced in \citet{DP20} in greater detail. The family of f-divergence induced risk measures is defined for each divergence generating function $\phi$ on domain($\phi$) = $[0,\infty)$ which is convex, lower semicontinuous, and satisfies $\phi(1)=0$. For two probability measures $P,Q$ on the sample space $\Omega$, the corresponding f-divergence generated by $\phi$ is defined as: 
\begin{equation}\label{f_divergence_def}
D_\phi(Q||P) = \int_{\Omega} \phi(\frac{dQ}{dP}) dP
\end{equation}
if $Q\ll P$ and $\infty$ otherwise.  The convex conjugate of $\phi$ is defined as $\psi(y) = \sup_{z\in \mathbb{R}} \{yz-\phi(z)\}$. These two functions satisfy the Fenchel–Young inequality among other important properties:
$$xy\leq \phi(x)+\psi(y)$$

For each pair of conjugate functions $\phi,\psi$ satisfying previous assumptions, the f-divergence induced risk measure is defined through the following minimization problem, where $\delta>0$ is a risk aversion parameter:
\begin{equation}\label{f_risk_p}
\rho_{\phi,\delta}(X) = \inf_{\mu\in \mathbb{R},t>0} t(\delta+\mu+\mathbb{E}[\psi(\frac{X}{t}-\mu)])
\end{equation}
The divergence function $\phi$ characterizes the shape
of risk aversion for increasing risk, while the risk aversion coefficient $\delta$ describes the tendency
of an investor to avoid risk. This risk measure is well defined for all $X\in L^1$, satisfying the inequality $\mathbb{E}[X]\leq\rho_{\phi,\delta}(X)$. It is law invariant but can be unbounded. As a remark, we know that $\phi(x)$ and $\tilde{\phi}(x) = \phi(x)+C(x-1),\forall C\in\mathbb{R}$ are both generating functions for the same divergence function due to conservation of probabilities. Their conjugate functions satisfy $\tilde{\psi}(x) = \psi(x+C)-C$, which leads to the same risk measure $\rho_{\phi,\delta}(X) = \rho_{\tilde{\phi},\delta}(X)$ from the primal problem \eqref{f_risk_p} for the same loss random variable $X$.

In Tables \ref{f_divergences} and \ref{f_conjugates} (tables adapted from \citet{N16}), we list some examples of f-divergences, including the well-known KL-divergence, total variation distance, and squared Hellinger distance. For each f-divergence listed, the generator function $\phi$ and its conjugate function $\psi$ are shown in the tables in accordance with notations of \citet{DP20}.

\begin{table}[!ht]
\renewcommand\thetable{1}
    \centering
      \caption{f-Divergences and Respective Generating Functions}
      \scalebox{0.9}{
       \begin{tabular}{|| c | c | c ||} 
 \midrule
Name & $D_\phi(P\|Q)$ & Generator $\phi(u)$ \\
\midrule
Total Variation & $\frac{1}{2}\int |p(x)-q(x)|dx$ & $\frac{1}{2}|u-1|$\\
\midrule
Kullback-Leibler & $\int p(x)\log{\frac{p(x)}{q(x)}}dx$ & $u\log{u}$\\
\midrule
Reverse Kullback-Leibler & $\int q(x)\log{\frac{q(x)}{p(x)}}dx$ & $-\log{u}$\\
\midrule
Pearson $\chi^2$ & $\int \frac{(p(x)-q(x))^2}{q(x)}dx$ & $(u-1)^2$\\
\midrule
Neyman $\chi^2$ & $\int \frac{(p(x)-q(x))^2}{p(x)}dx$ & $\frac{(u-1)^2}{u}$\\
\midrule
Squared Hellinger & $\int (\sqrt{p(x)}-\sqrt{q(x)})^2 dx$ & $(\sqrt{u}-1)^2$\\
\midrule
Jensen-Shannon & $\frac{1}{2}\int (p(x)\log{\frac{2p(x)}{p(x)+q(x)}}+q(x)\log{\frac{2q(x)}{p(x)+q(x)}})dx$ & $-(u+1)\log{\frac{1+u}{2}}+u\log{u}$\\
\midrule
$\alpha$-Divergence ($\alpha\notin \{0,1\}$) & $\frac{1}{\alpha(\alpha-1)}\int (p(x)(\frac{q(x)}{p(x)})^\alpha-\alpha q(x)-(1-\alpha)p(x))dx$ & $\frac{1}{\alpha(\alpha-1)}(u^\alpha -1-\alpha(u-1))$\\
\midrule
\end{tabular}}
\label{f_divergences}
\end{table}

\begin{table}[!ht]
    \centering
      \caption{f-Divergence Generating Function Domains and Conjugate Functions}
      \scalebox{1}{
       \begin{tabular}{|| c | c | c ||} 
\midrule
Name & $dom(\psi)$ & Conjugate $\psi(t)$ \\
\midrule
Total Variation & $-\frac{1}{2}\leq t\leq\frac{1}{2}$ & $t$\\
\midrule
Kullback-Leibler & $\mathbb{R}$ & $e^{t-1}$\\
\midrule
Reverse Kullback-Leibler & $\mathbb{R}_{-}$ & $-1-\log{(-t)}$\\
\midrule
Pearson $\chi^2$ & $\mathbb{R}$ & $\frac{1}{4}t^2+t$\\
\midrule
Neyman $\chi^2$ & $t<1$ & $2-2\sqrt{1-t}$\\
\midrule
Squared Hellinger & $t<1$ & $\frac{t}{1-t}$\\
\midrule
Jensen-Shannon & $t<\log{2}$ & $-\log{(2-e^{t})}$\\
\midrule
$\alpha$-Divergence ($\alpha<1,\alpha\neq 0$) & $t<\frac{1}{1-\alpha}$ & $\frac{1}{\alpha}(t(\alpha-1)+1)^{\frac{\alpha}{\alpha-1}} - \frac{1}{\alpha}$\\
\midrule
$\alpha$-Divergence ($\alpha>1$) & $\mathbb{R}$ & $\frac{1}{\alpha}(t(\alpha-1)+1)^{\frac{\alpha}{\alpha-1}} - \frac{1}{\alpha}$\\
\midrule
\end{tabular}}
\label{f_conjugates}
\end{table}

The f-divergence induced risk measure $\rho_{\phi,\delta}: L^1\to R\cup\{\infty\}$ satisfies the following four axioms:

1. Monotonicity: $\rho(X_1) \leq \rho(X_2)$ provided that $X_1 \leq X_2$ almost surely.

2. Translation equivariance: $\rho(X + c) = \rho(X) + c$ for any $X\in L^1$ and $c \in \mathbb{R}$.

3. Subadditivity: $\rho(X_1 + X_2) \leq \rho(X_1) + \rho(X_2)$ for all $X_1, X_2 \in L^1$.

4. Positive homogeneity: $\rho(\lambda X) = \lambda \rho(X)$ for all $X \in L^1$ and $\lambda > 0$.

Hence it is a coherent risk measure by definition. It also satisfies $\forall 0<\delta_1\leq\delta_2$, 
$$\rho_{\phi,\delta_1}(X)\leq \rho_{\phi,\delta_2}(X), \forall X\in L^1$$
Conversely, for any non-negative random variable $X\geq 0$, it holds that, $$\rho_{\phi,\delta_2}(X)\leq \frac{\delta_2}{\delta_1}\rho_{\phi,\delta_1}(X)$$

The most important property of the f-divergence induced risk measures is that they can be equivalently stated through a dual representation involving a maximization problem. For a given divergence generating function $\phi$ and its convex conjugate $\psi$, define the Orlicz space $L^\psi$ as $L^\psi = \{X\in L^0: \mathbb{E}[\psi(t|X|)] < \infty$ for some $t > 0\}$. Dommel and Pichler \cite{DP20} showed that for every $X\in L^{\psi}$, the f-divergence induced risk measure has the dual representation:
\begin{equation}\label{f_risk_dual}
\rho_{\phi,\delta}(X) = \sup_{Z\in M_{\phi,\delta}} \mathbb{E}[XZ]
\end{equation}
where $M_{\phi,\delta} = \{Z\in L^1:Z\geq 0,\mathbb{E}[Z]=1,\mathbb{E}[\phi(Z)]\leq\delta\}$. The above dual representation in \eqref{f_risk_dual} can be alternatively expressed as:

\begin{equation}\label{f_risk_dual2}
\rho_{\phi,\delta}(X) = \sup_{Q\ll P: D_\phi(Q||P)\leq\delta} \mathbb{E}_Q[X]
\end{equation}
The proposed risk can therefore be interpreted
as the largest expected value over all probability measures $Q$ within an f-divergence ball around $P$. The divergence function $\phi$ characterizes the shape of the ball, while $\delta$ determines the radius. Under suitable conditions where $\bar{\alpha} = \max\{\alpha\in[0,1): \phi(0)\alpha+\phi(\frac{1}{1-\alpha})(1-\alpha)\leq\delta\}$ and $\mathbb{P}(X=ess sup(X))<1-\bar{\alpha}$ holds true, then the infimum in the defining equation \eqref{f_risk_p} of the risk measure is attained. Again, one can make the remark that for generating functions $\phi,\tilde{\phi}$ which generate the same divergence, the resulting risk measure is the same since the feasible region in \eqref{f_risk_dual2} is the same. As an additional remark, Entropic Value-at-Risk (EVaR), see \cite{A12} and \cite{AF19}, is a frequently used coherent risk measure in finance and engineering, defined as:
\begin{equation}\label{EVaR}
EVaR_{1-\alpha}(X)= \sup_{Q\ll P: D_{KL}(Q||P)\leq -\ln{\alpha}} \mathbb{E}_Q[X]
\end{equation}
EVaR is an example of the family of f-divergence induced risk measures by choosing the Kullback-Leibler divergence (KLD) to be the f-divergence and the confidence radius to be $\delta:=-\ln{\alpha}$ for some confidence level $1-\alpha\in[0,1)$. When $X\sim N(\mu,\sigma^2)$ is a univariate Gaussian distributed loss random variable, EVaR takes the analytical form $EVaR_{1-\alpha}(X) = \mu+\sigma\sqrt{-2\ln{\alpha}}$. Similarly, an analytic formula can be computed for a uniform random variable.

The f-divergence induced risk measure also has the representation in terms of spectral risk measures \cite{DP20}, which is equivalent to the Kusuoka representation:
\begin{equation}\label{kusuoka}
\rho_{\phi,\delta}(X) = \sup_{\sigma} \int_0^1 \sigma(u)F^{-1}_X(u)du
\end{equation}
where the supremum is taken over all non-decreasing $\sigma:[0,1]\to[0,\infty]$ with $\int_0^1 \sigma(u)du=1$ and $\int_0^1 \phi(\sigma(u))du\leq\delta$. Notice that every function of the form $\rho_\sigma(X)= \int_0^1 \sigma(u)F^{-1}_X(u)du$ where $\sigma:[0,1]\to[0,\infty]$ with $\int_0^1 \sigma(u)du=1$ is a coherent risk measure itself, hence the spectral representation is useful for deriving bounds between the f-divergence induced risk measure and its spectral components. An important special example is the CVaR risk functional, which has the following spectral form:
$$CVaR_\alpha(X) = \rho_{\sigma_\alpha}(X) = \int_0^1 \sigma_\alpha(u)F_X^{-1}(u)du = \frac{1}{1-\alpha}\int_{\alpha}^1 F_X^{-1}(u)du, \forall \alpha\in(0,1)$$
where the function $\sigma_\alpha(u) = \frac{1}{1-\alpha}\mathbbm{1}_{[\alpha,1]}(u)$
is in the spectral set of the f-divergence induced risk measure (over which the supremum is taken) if the following holds: $\phi(0)\alpha+\phi(\frac{1}{1-\alpha})(1-\alpha)\leq\delta$. Hence under this condition, we can relate the corresponding CVaR risk to the f-divergence induced risk by the inequality:
$$CVaR_\alpha(X)\leq\rho_{\phi,\delta}(X),\forall X\in L^\psi$$
Dommel and Pichler \cite{DP20} elaborate on the optimality inside of the minimization representation \eqref{f_risk_p} and the dual maximization representation \eqref{f_risk_dual} based on facts concerning the ‘derivatives’ of the convex function $\phi$ and its conjugate $\psi$. 
Let $X\in L^\psi$ and suppose $t^\star,\mu^\star$ solve the characterizing equations:
\begin{equation}\label{chEq}
\begin{aligned}
\mathbb{E}[\psi'(\frac{X}{t}-\mu)]=1; \\
\mathbb{E}[\phi(\psi'(\frac{X}{t}-\mu))] = \delta
\end{aligned}
\end{equation}
Then they are the optimal values in the primal minimization problem. Furthermore, the random variable $Z^\star= \psi'(\frac{X}{t^\star}-\mu^\star)$ is optimal in the dual maximization problem, i.e., $Z^\star\in M_{\phi,\delta}$ and,
$$\sup_{Z\in M_{\phi,\delta}}\mathbb{E}[XZ] = \mathbb{E}[XZ^\star] = t^\star(\delta+\mu^\star+\mathbb{E}[\psi(\frac{X}{t^\star}-\mu^\star)])$$

Lastly, the primal formulation of the f-divergence induced risk measure can be efficiently used in portfolio optimization problem objectives where we can solve a single augmented minimization problem with only two additional variables instead of two nested optimization problems when using the dual maximization problem representation or its alternative forms. With the primal formulation, we solve for:
$$\min_{x\in X} \rho_{\phi,\delta}(\mathcal{L}(x)) = \min_{x\in X} \inf_{\mu\in \mathbb{R},t>0} t(\delta+\mu+\mathbb{E}[\psi(\frac{\mathcal{L}(x)}{t}-\mu)])$$
Here the portfolio loss random variable is $\mathcal{L}(x) = -x^T R$ where $R$ is a random vector of returns for the $N$ assets in consideration with expected values $r$. The feasible region $X$ can be chosen using standard constraints on the (long-only) portfolio weights in combination with a target return constraint with threshold $\bar{r}$ such as:
$$X = \{x\in\mathbb{R}^N: \mathbbm{1}^T x = 1,x\geq 0,r^T x\geq \bar{r}\}$$
Similar to the discussion before, $\phi$ controls the shape of the risk function and hence the sensitivity of the solution to the return distribution and the desired risk level $\delta$. Larger $\delta$ implies higher risk aversion and more conservative portfolio choices.

In the following sections, we can study CAPM-type necessary optimality conditions of optimization problems based on $\rho_{\phi,\delta}$ risk objectives and derive corresponding Beta values for assets in a market under CAPM assumptions, which generalizes the Standard Beta and provides an alternative to CVaR Betas and CVaR-based Drawdown Betas as discussed in \citet{DU22}, \cite{13}.

\section{f-Divergence induced Risk Measure and Portfolio Optimization}

Let us denote by $r(x)$ the vector of historic returns of a portfolio with weights vector $x$. Also, we  denote by
$l(x) = -r(x)$ the associated loss observations.

We denote
\begin{itemize}
\item 
$x= (x^1, \ldots,x^I)\;=\;$   vector of weights for $I$ assets in the portfolio;
\item
$(r_{st}^1,\ldots,r_{st}^I)\;=\;$ vector of returns of portfolio assets at time moment $t$ on scenario $s$;
\item
$(l_{st}^1,\ldots,l_{st}^I)\;=\;(-r_{st}^1,\ldots,-r_{st}^I)\;=\;$ corresponding vector of losses of portfolio assets;
\item
$p_s\;=\;$ probability of the scenario (sample path of returns of securities);
\item
$r_{st}(x)= \sum _{i=1}^{I} r_{st}^i x^i\;=\;$ portfolio return at time  moment $t$ on scenario $s$;
\item
$r(x)=$ vector of portfolio returns with components
$r_{st}(x)\,$, $\,s=1,\ldots,S;\, t=1,\ldots,T\,$;
\item
$l_{st}(x)= -r_{st}(x) =\;$ portfolio loss at time  moment $t$ on scenario $s$;

\end{itemize}
Following \citet{13}, we state the $\rho_{\phi,\delta}$ risk (multiple paths) minimization problem over $T$ periods subject to a constraint that the portfolio's expected return exceeds a given target $\bar{r}$:
\begin{equation}
\label{fdrmin}
\begin{aligned}
& &\underset{x}{\text{min}}
& & \rho_{\phi,\delta}(l(x)) 
& & s.t.\;\; \frac{1}{T}\sum_{s=1}^{S}p_s\sum_{t=1}^T l_{st}(x) \leq -\bar{r}\;.
\end{aligned}
\end{equation}
This problem is similar to a Markowitz mean-variance optimization  with variance replaced by the f-divergence induced risk $\rho_{\phi,\delta}$.

The above minimization problem \eqref{fdrmin} is equivalent to the maximization problem below, where a given target $\nu$ specifies a constraint on the portfolio risk,
\begin{equation}
\label{fdrmax}
\begin{aligned}
& &\underset{x}{\text{max}}
& & -\frac{1}{T}\sum_{s=1}^{S}p_s\sum_{t=1}^T l_{st}(x)
& & s.t.\;\;\;\rho_{\phi,\delta}(l(x))\leq v \;,
\end{aligned}
\end{equation}
in the sense that the efficient frontiers of these two problems \eqref{fdrmin} and (\ref{fdrmax}) coincide.

\section{CAPM: Necessary Optimality Conditions for Portfolio Optimization with f-Divergence induced Risk Measures}
Similar to \citet{13} and \citet{DU22}, we provide necessary optimality conditions for optimization problems \eqref{fdrmin} and (\ref{fdrmax})
in the form of CAPM equations. In particular, the formula for f-Beta can be derived similarly to the Standard Beta, which relates returns of the market and individual assets. Let $l^M$ be the loss associated with the optimal(market) portfolio, the necessary optimality conditions for the solution $x^\star$ of both problems \eqref{fdrmin} and (\ref{fdrmax}) are stated in the form of CAPM:
\begin{equation}
\label{CAPMeq}
\frac{1}{T}\sum_{s=1}^{S} p_s \sum_{t=1}^T r_{st}^i = \beta_{\rho,\delta}^i\frac{1}{T}\sum_{s=1}^{S}p_s\sum_{t=1}^T r_{st}^M\;,
\end{equation}
\begin{equation}
\label{beta}
\beta_{\rho,\delta}^i = \frac{\sum_{s=1}^{S}\sum_{t=1}^{T}p_s q_{st}^\star l_{st}^i}{\rho_{\phi,\delta}(l^M)}\;,
\end{equation}

The f-Beta CAPM equation (\ref{CAPMeq}) relates the expected returns of the market and instruments under a constrained perturbation measure. On the efficient frontier with the $\rho_{\phi,\delta}$ risk measure against the target return, the optimal solution $x^\star$ is the point where the capital asset line makes a tangent cut with the efficient frontier.

\begin{theorem}
\label{thm1}
\noindent
Let $r^M = r(x^\star)$ be the return vector for an optimal portfolio $x^\star$ of the problem \eqref{fdrmin}.
The necessary optimality conditions for \eqref{fdrmin} can be stated in the form of CAPM:

$$\frac{1}{T}\sum_{s=1}^{S} p_s \sum_{t=1}^T r_{st}^i = \beta_{\rho,\delta}^i \frac{1}{T}\sum_{s=1}^{S}p_s\sum_{t=1}^T r_{st}^M\;,$$

\begin{equation}
\label{fbeta2}
\beta_{\rho,\delta}^i = \frac{\sum_{s=1}^{S}\sum_{t=1}^{T}p_s q_{st}^\star l_{st}^i}{\rho_{\phi,\delta}(l^M)}\;,
\end{equation}
where
\begin{itemize}

\item $\beta_{\phi,\delta}^i\;=\;$ f-Beta relating the expected returns,
$\frac{1}{T}\sum_{s=1}^{S}p_s \sum_{t=1}^T r_{st}^M$,
of the optimal portfolio (market) and expected return, $\frac{1}{T}\sum_{s=1}^{S} p_s \sum_{t=1}^T r_{st}^i$, of security $i$;

\item $l^M = \{l^M_{st}\}_{s,t=1}^{S,T}$ are the loss observations associated with the optimal(market) portfolio for all $s=1,\ldots,S,t=1,\ldots,T$;

\item $\rho_{\phi,\delta}(l^M) = \max_{\{q_{st}\}_{s,t=1}^{S,T}\in\mathcal{Q}_{\phi,\delta}}\sum_{s=1}^S\sum_{t=1}^T p_s q_{st} l^M_{st}$ is the f-divergence induced risk for the optimal(market) portfolio, where $\mathcal{Q}_{\phi,\delta} = \{\{q_{st}\}_{s,t=1}^{S,T}|\sum_{s=1}^S\sum_{t=1}^T p_s q_{st}=1,\sum_{s=1}^S p_s D_{\phi}(q_s||\bar{p})\leq\delta\}$ is the set of feasible perturbation measures under f-divergence function $\phi$ and confidence radius $\delta$, where $\bar{p}_t = \frac{1}{T},\forall t=1,\ldots,T$ is the empirical probability vector for $T$ observations in a single historic path;

\item $q^\star_{st} = $ optimal perturbed measure from $\mathcal{Q}_{\phi,\delta}$ which maximizes the perturbed market expected loss as in the dual representation of $\rho_{\phi,\delta}(l^M)$;
\end{itemize}
Consequently, for a single path: 
$$
\frac{1}{T}\sum_{t=1}^T r^i_t \;=\; \beta_{\phi,\delta}^i\, \frac{1}{T}\sum_{t=1}^T r_t^M\;,
$$
\begin{equation}
\label{fbeta_3}
\beta_{\phi,\delta}^i\; =\;\frac{\sum_{t=1}^{T}q_{t}^\star l_t^i }{\rho_{\phi,\delta}(l^M)}\;,
\end{equation}

where
\begin{itemize}
\item $\beta_{\phi,\delta}^i\;=\;$ f-Beta relating the expected returns,
$\frac{1}{T}\sum_{t=1}^T r_{t}^M$,
of the optimal portfolio (market) and expected return, $\frac{1}{T}\sum_{t=1}^T r_{t}^i$, of security $i$ in a single path;
\item
$\rho_{\phi,\delta}(l^M) = \max_{\{q_{t}\}_{t=1}^{T}\in\mathcal{Q}_{\phi,\delta}}\sum_{t=1}^T q_{t} l^M_{t}$ is the f-divergence induced risk for the optimal(market) portfolio, where $\mathcal{Q}_{\phi,\delta} = \{\{q_{t}\}_{t=1}^{T}|\sum_{t=1}^T q_{t}=1,D_{\phi}(q||\bar{p})\leq\delta\}$ is the set of feasible perturbation measures under f-divergence function $\phi$ and confidence radius $\delta$, where $\bar{p}_t = \frac{1}{T},\forall t=1,\ldots,T$ is the empirical probability vector for $T$ observations in a single historic path;
\item
$q^\star_{t}$ = optimal perturbed measure from $\mathcal{Q}_{\phi,\delta}$ which maximizes the perturbed market expected loss as in the dual representation of $\rho_{\phi,\delta}(l^M)$ for a single path;   

\end{itemize}
\end{theorem}

\begin{proof}
Let us denote by $\rho_{\phi,\delta}(l(x)) = \max_{\{q_{st}\}_{s,t=1}^{S,T}\in\mathcal{Q}_{\phi,\delta}} \sum_{s=1}^{S}\sum_{t=1}^{T} p_s q_{st} l_{st}(x)$ the objective function and by $x^\star$ an optimal portfolio vector of problem \eqref{fdrmin}.
The expected loss function $l_{st}(x)$ linearly depends upon vector $x$. The objective function $\rho_{\phi,\delta}(l(x))$ is convex in $x$. 
The necessary optimality condition for the convex optimization problem \eqref{fdrmin} is formulated as follows (see, for reference, Theorem 3.34 in \citet{16}):
\begin{equation}
\label{neces}
\lambda^\star\, \nabla_{\!x}  (\frac{1}{T}\sum_{s=1}^{S}p_s \sum_{t=1}^T r_{st}(x^\star)-\bar{r})\;\in\; \partial_x \rho_{\phi,\delta}(l(x^\star)) \;,
\end{equation}
where
\begin{itemize}
\item 
$\nabla_{\!x}  (\frac{1}{T}\sum_{s=1}^{S}p_s\sum_{t=1}^T r_{st}(x^\star)-\bar{r})\;=\;$ gradient of the constraint function at $x=x^\star$;
\item
$\lambda^\star\;=\;$ Lagrange multiplier such that $\lambda^\star\geq 0$ and $\lambda^\star(\frac{1}{T}\sum_{s=1}^S\sum_{t=1}^T p_s r_{st}(x^\star)-\bar{r})=0\;$;
\item
$\partial_x \rho_{\phi,\delta}(l(x^\star)) \;=\;$ subdifferential of convex in $x$ function $\rho_{\phi,\delta}(l(x))$  at $x=x^\star$.
\end{itemize}
The gradient of the constraint function, which is linear in $x$, equals:
\begin{equation}
\label{grfdr}
\nabla_{\!x}  (\frac{1}{T}\sum_{s=1}^{S}p_s \sum_{t=1}^T r_{st}(x^\star)-\bar{r})= 
 \frac{1}{T}\sum_{s=1}^S \sum_{t=1}^T p_s \nabla_{\!x}  r_{st}(x)=
 \sum_{s=1}^S p_s(\frac{1}{T}\sum_{t=1}^T r^1_{st},\ldots,\frac{1}{T}\sum_{t=1}^T r^n_{st})\;.
\end{equation} 
According to the standard results in convex analysis,
\begin{equation}
\label{inclu}
g = (g^1,\ldots,g^n) \in \partial_x \rho_{\phi,\delta}(l(x^\star))\;,\;\;\; \text{where}\;\;\;
g^i= \sum_{s=1}^S\sum_{t=1}^{T}p_s q_{st}^\star l^i_{st} \;.
\end{equation}
With (\ref{grfdr}) and (\ref{inclu}) we obtain the following system of equations
\begin{equation}
\label{equ}
g^i= \sum_{s=1}^S\sum_{t=1}^{T}p_s q_{st}^\star l^i_{st} =\lambda^\star \frac{1}{T}\sum_{s=1}^S\sum_{t=1}^T p_s r^i_{st}\;,\;\;\;i=1,\ldots,I \;.
\end{equation}
Multiplying the left and right hand sides of the previous equation by the optimal $x^\star_i$ and summing up terms for $i=1,\ldots,I$, we have:
$$
\sum_{s=1}^S\sum_{t=1}^{T}p_s q_{st}^\star l_{st}(x^\star) =\lambda^\star \frac{1}{T}\sum_{s=1}^S\sum_{t=1}^T p_s r_{st}(x^\star)\;.
$$
Consequently,
\begin{equation}
\label{equ1}
\lambda^\star=\frac{\sum_{s=1}^S\sum_{t=1}^{T}p_s q_{st}^\star l_{st}(x^\star)}{\frac{1}{T}\sum_{s=1}^S p_s \sum_{t=1}^T r_{st}(x^\star)}\;.
\end{equation}
Substituting
(\ref{equ1}) to (\ref{equ})
gives necessary conditions (\ref{CAPMeq}), (\ref{fbeta2})
in CAPM format, where we denote the optimal portfolio as the market portfolio $l(x^\star) = l^M,r(x^\star) = r^M$, and the optimal risk is $\rho_{\phi,\delta}(x^\star) =\sum_{s=1}^S\sum_{t=1}^{T}p_s q_{st}^\star l_{st}(x^\star) $.
\end{proof}

\section{Discussion of f-Divergence induced Risk and f-Betas}
Compared with other coherent risk measures such as CVaR, f-divergence induced risk measures have the advantage that it is the result of a well-defined distributionally robust optimization problem under probability model uncertainty as specified through information-theoretic distances or divergences. The risk measured in this formulation has a direct meaning related to the choice of f-divergence functions.

The flexibility in the choices of $\phi,\delta$ allows for a spectrum of solutions that are comparable and interpretable. In particular, the relations between $\rho_{\phi,\delta}$ and CVaR gives a direct comparison between Betas calculated based on these two risk measures, respectively.

In the following section, we use a particular f-divergence function to illustrate the behavior of the proposed Beta metrics. Consider the following symmetric f-divergence called the squared Hellinger distance (\citet{YL}) for two univariate probability distributions $P,Q$ with densities $p,q$:
\begin{equation}\label{h2}
D_{H^2}(Q||P) = D_{H^2}(P||Q) := H^2(P,Q)=\frac{1}{2}\int_x (\sqrt{q(x)}-\sqrt{p(x)})^2 dx
\end{equation}
Squared Hellinger distance is related to the total variation distance, an important symmetric f-divergence, via the following inequalities, $H^2(P,Q)\leq D_{TV}(P,Q)\leq\sqrt{2}H(P,Q)$,
where $D_{TV}$ is defined as,
\begin{equation}\label{tvd}
D_{TV}(P,Q) = \frac{1}{2}\int_x |p(x)-q(x)|dx
\end{equation}
Squared Hellinger distance is also closely related to Kullback-Leibler divergence (\citet{KL}) and can be bounded by
$2H^2(P,Q)\leq D_{KL}(P||Q)$, where $D_{KL}$ is defined as,
\begin{equation}\label{kld}
D_{KL}(P||Q) = \int_x p(x)\log\frac{p(x)}{q(x)}dx
\end{equation}

The squared Hellinger distance \eqref{h2} belongs to a one-parameter family of f-divergences called the $\alpha$-divergences (\citet{CA10}), which smoothly connects well-known divergences such as the KL divergence, reverse KL divergence, and Pearson $\chi^2$ divergence by varying the $\alpha$ parameter \cite{BP18}:
\begin{equation}\label{alphaD}
D_A^\alpha(P||Q) = \frac{1}{\alpha}+\frac{1}{1-\alpha} - \frac{1}{\alpha(1-\alpha)}\int_x (\frac{q(x)}{p(x)})^{1-\alpha} p(x)dx
\end{equation}
Among its other properties, the squared Hellinger distance (up to a constant scaling factor) has been shown to be the minimum distance within a symmetrized $\alpha$-divergence family \cite{M22}.

For a given historic path of length $T$, the nominal probability vector is $\bar{p}=(\frac{1}{T},\ldots,\frac{1}{T})$ with length $T$. Hence the squared Hellinger distance for any perturbation measure $Q$ is bounded between $[0,1-\sqrt{\frac{1}{T}}]$. The Hellinger-Beta is proposed as follows (for a single path):

\begin{equation}
\label{hbeta}
\beta_{H^2,\delta}^i = \frac{\sum_{t=1}^{T}q_{t}^\star l_{t}^i}{\rho_{H^2,\delta}(l^M)}\;,
\end{equation}

\begin{itemize}
\item 
$\rho_{H^2,\delta}(l^M) = \max_{\{q_{t}\}_{t=1}^{T}\in\mathcal{Q}_{H^2,\delta}}\sum_{t=1}^T q_{t} l^M_{t}$ is the squared Hellinger distance induced risk for the optimal(market) portfolio, where $\mathcal{Q}_{H^2,\delta} = \{\{q_{t}\}_{t=1}^{T}|\sum_{t=1}^T q_{t}=1,D_{H^2}(q||\bar{p})\leq\delta\}$ is the set of feasible perturbation measures, where $\bar{p}_t = \frac{1}{T},\forall t=1,\ldots,T$ is the empirical probability vector for $T$ observations in a single historic path;
\item
$q^\star_{t}$ = optimal perturbed measure from $\mathcal{Q}_{H^2,\delta}$ which maximizes the perturbed market expected loss as in the dual representation of $\rho_{H^2,\delta}(l^M)$ for a single path; 
\end{itemize}

The optimal perturbation measure of the inner maximization problem of perturbed loss can be efficiently solved using convex optimization procedures, from which the Hellinger-Beta calculation follows easily. The squared Hellinger distance possesses desirable properties such as symmetry, boundedness between 0 and 1, and provides lower and upper bounds for total variation distance while varying smoothly as the probability vectors change. The boundedness of squared Hellinger distance gives a clear range of choice for risk aversion parameter $\delta$ when designing Beta metrics based on its induced risk measure.

\begin{figure}[ht!]
    \centering
    \subfloat[\centering $CVaR_\alpha$ with tail risk level $1-\alpha$]{{\includegraphics[width=7.5cm]{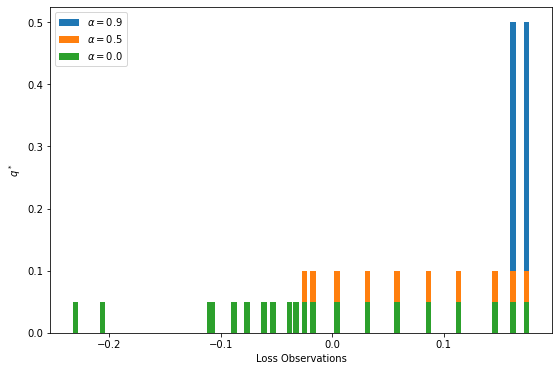} }}%
    \quad
    \subfloat[\centering $\rho_{H^2,\delta}$ with divergence radius $\delta$]{{\includegraphics[width=7.5cm]{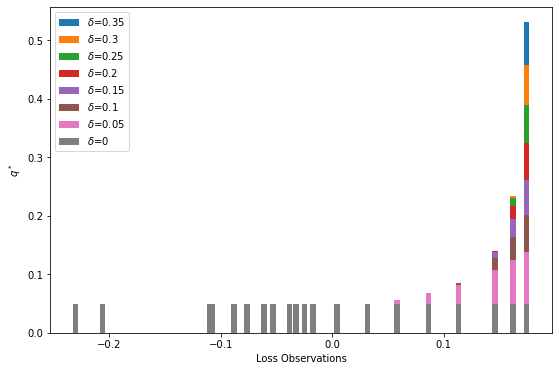} }}%
    \caption{Comparison of $\rho_{H^2,\delta}$ and $CVaR_\alpha$ Optimal Risk Identifiers}%
    \label{fig:f_q_demo}%
\end{figure}

As a comparison of the different behavior of f-divergence (and in particular squared Hellinger distance) induced risk measures and CVaR risks, Figure \ref{fig:f_q_demo} shows the risk identifiers $q^\star$ of $\rho_{H^2,\delta}$ for 20 loss observations sampled from a Gaussian distribution, where $\delta$ is varying, compared against $CVaR_\alpha$ risk identifiers which are scaled indicator functions (step functions) on the $(1-\alpha)\times 100\%$ tail for various $\alpha$. This shows the different risk behaviors offered by the family of f-divergence induced risk measures, which generally have smoother transitions of the risk identifiers than the tail indicator functions for CVaR. In particular, large $\delta$ corresponds to worst-case loss, which is similar to large $\alpha\to 1$, while $\delta=0$ and $\alpha=0$ both reduce to the risk-neutral expectation. The behavior for intermediate values of $\delta>0$ is not directly comparable with intermediate values of $\alpha\in(0,1)$. Notice that for an empirical distribution of loss observations and for any fixed $\delta>0$, the optimal risk identifiers for $\rho_{H^2,\delta}$ are always larger for larger (more severe) loss observations. This result can be formalized as follows, where $\bar{p}_t = \frac{1}{T},\forall t=1,\ldots,T$ is the probability vector associated with the empirical loss distribution $l = \{l_t\}_{t=1}^T$ based on $T$ loss observations. 

\begin{lemma}\label{increasing_q}
Consider the dual maximization problem for the f-divergence induced risk measures of an empirical loss distribution $l = \{l_t\}_{t=1}^T$ with a given divergence generating function $\phi$ and risk aversion radius $\delta>0$:
\begin{equation}\label{dual_opt_empirical}
\rho_{\phi,\delta}(l) = \max_{\{q_{t}\}_{t=1}^{T}\in\mathcal{Q}_{\phi,\delta}}\sum_{t=1}^T q_{t} l_{t}    
\end{equation} 
where $\mathcal{Q}_{\phi,\delta} = \{\{q_{t}\}_{t=1}^{T}|\sum_{t=1}^T q_{t}=1,D_{\phi}(q||\bar{p})\leq\delta\}$ and $D_\phi(q||\bar{p}) = \sum_{t=1}^T \phi(\frac{q_t}{\bar{p}_t})\bar{p}_t = \frac{1}{T}\sum_{t=1}^T \phi(T q_t)$. Let $\{q^\star_t\}_{t=1}^T$ be an optimal solution to \eqref{dual_opt_empirical}. Then for any $\forall i,j\in\{1,\ldots,T\}$, $i\neq j$ such that the loss observations $l_i>l_j$, we must have $q^\star_i\geq q^\star_j$ in the optimal solution.

\end{lemma}
\begin{proof}
Suppose $\{q^\star_t\}_{t=1}^T$ is an optimal solution for \eqref{dual_opt_empirical}, where there exists $i,j\in\{1,\ldots,T\},i\neq j$, such that $l_i>l_j$ and $q^\star_i<q^\star_j$. Consider the alternative solution $\{q'_t\}_{t=1}^T$ such that $q'_t = q^\star_t,\forall t\in\{1,\ldots,T\}$ such that $t\neq i$ and $t\neq j$, $q'_i = q^\star_j$, and $q'_j = q^\star_i$. $\{q'_t\}_{t=1}^T$ is a feasible solution to \eqref{dual_opt_empirical} because $D_\phi(q'||\bar{p}) =\frac{1}{T}\sum_{t=1}^T \phi(T q'_t) =\frac{1}{T}\sum_{t=1}^T \phi(T q^\star_t)\leq \delta$ by assumption. However, this solution $\{q'_t\}_{t=1}^T$ strictly increases the objective because, 
$$\sum_{t=1}^T q'_t l_t = \sum_{t=1}^T q^\star_t l_t+ (q'_i- q^\star_i) l_i+ (q'_j - q^\star_j) l_j = \sum_{t=1}^T q^\star_t l_t + (q^\star_j - q^\star_i)(l_i-l_j) $$
where by our assumptions $l_i>l_j$ and $q^\star_i<q^\star_j$, hence $\sum_{t=1}^T q'_t l_t > \sum_{t=1}^T q^\star_t l_t$. This raises a contradiction since $\{q^\star_t\}_{t=1}^T$ is assumed to be an optimal solution. This completes the proof.
\end{proof}
Lemma \ref{increasing_q} formalizes the behavior of the risk identifiers for f-divergence induced risk measures, as demonstrated in Figure \ref{fig:f_q_demo} for the case of squared Hellinger distance in particular. Hence we expect such risk measures to put a heavier focus on more extreme losses, which is natural for a risk measure. For empirical loss observations (which are bounded), the optimal objective value of \eqref{dual_opt_empirical} is always greater than the empirical mean but no larger than the most extreme loss: $\frac{1}{T}\sum_{t=1}^T l_t\leq\rho_{\phi,\delta}(l)\leq \max_{t=1,\ldots,T} l_t$.

As a last remark, the proposed f-Betas are also related to a more general theory about deviation-based Betas, see \citet{14}. The particular Beta in the deviation-based framework would be (for a single path):
\begin{equation}
\label{hbeta_dev}
\beta_{H^2,\delta}^{i,dev} = \frac{\sum_{t=1}^{T}q_{t}^\star d_{t}^i}{\rho_{H^2,\delta}(d^M)}\;,
\end{equation}
where the deviation measures are defined as $d^i_t = \mathbb{E}[r^i]-r_t^i$ and $d^M = \mathbb{E}[r^M] - r^M$ and the risk identifiers $q_t^\star$ are defined similarly as before.
\section{Numerical Experiments}

In the following numerical experiments, we computed Beta values using the S\&P 500 index as the chosen optimal (market) portfolio. For demonstration, only the resulting values for the DOW 30 stocks are presented. To evaluate the stability over time of the different Betas, we first considered
two non-overlapping 7.5-year historic periods:
\begin{itemize}
    \item Period 1 (containing Financial Crisis 2008): 
from 2006/01/01 to 2013/06/28; \item  Period 2 (containing COVID-19 Crisis):  from 2013/07/01 to 2021/01/01. 
\end{itemize}

\begin{table}[!ht]
    \centering
      \caption{Betas for DOW 30 Stocks: Non-overlapping Period 1 and 2 }
      \scalebox{0.85}{
       \begin{tabular}{|| c | c | c | c | c | c | c | c | c | c | c||} 
 \hline
 & $\rho^{dev}_{H^2,0.2}$ & $\rho^{dev}_{H^2,0.2}$ & $\rho_{H^2,0.2}$ & $\rho_{H^2,0.2}$ & $ERoD_{0+}$ & $ERoD_{0+}$ & $CDaR_{0.5}$
 & $CDaR_{0.5}$ & Standard & Standard\\  

  &  Period 1 & Period 2 &  Period 1 & Period 2 & Period 1  & Period 2 & Period 1
 & Period 2 & Period 1 & Period 2 \\
 \hline
  AAPL & 0.861 & 0.868 & 0.835 & 0.828 & -0.986 & 1.126  & -0.796 & 1.207 & 0.894 & 1.064\\ 
 \hline
AMGN & 0.650 & 0.663 & 0.647 & 0.648 & -0.218 & 0.468  & -0.163 & 0.523 & 0.817 & 0.954 \\
 \hline
 AXP & 1.648 & 1.260 & 1.641 & 1.266& 0.540 & 1.606  & 0.659 & 1.547  & 1.601 & 1.056\\
 \hline
 BA & 0.702 & 2.613 & 0.693 & 2.637 & 1.332 & 1.277 & 1.376 & 1.342 & 0.866 & 1.365 \\
 \hline
 CAT & 1.148 & 1.329 &1.141 & 1.328 & 0.558 & 1.892 & 0.626 & 1.868 & 1.252 & 1.406 \\
 \hline
 CRM & 1.224 & 0.964 & 1.198 & 0.935 & -1.438 & 0.154 & -1.268 & 0.207 & 1.533 & 1.170\\
 \hline
 CSCO & 1.091 & 1.234 & 1.086 & 1.236 & 1.119 & 0.664 & 1.015 & 0.641 & 1.116 & 1.245 \\
 \hline
 CVX & 1.225 & 1.442 & 1.214 & 1.467 & -0.129 & 1.323 & -0.038 & 1.303 & 1.042 & 1.259 \\
 \hline
 DIS & 1.038 & 1.270 & 1.025 & 1.267 & 0.081 & 0.859 & 0.197 & 0.893 & 0.951 & 1.025 \\
 \hline
 GS & 1.304 & 1.518 & 1.298 & 1.529 & 0.640 & 1.956 & 0.564 & 1.896 & 1.445 & 1.325 \\
 \hline
 HD & 0.755 & 1.026 & 0.743 & 1.012 & -0.037 & 0.253 & 0.040 & 0.373 & 1.003 & 0.901 \\
 \hline
 HON & 0.957 & 1.173 & 0.944 & 1.168 & 0.768 & 0.967 & 0.830 & 1.015 & 1.026 & 1.105 \\
 \hline
 IBM & 0.617 & 1.424 & 0.604 & 1.453 & -0.410 & 1.986 & -0.337 & 1.901 & 0.721 & 1.002 \\
 \hline
 INTC & 0.860 & 1.095 & 0.857 & 1.090 & 0.317 & 0.405 & 0.323 & 0.386 & 0.776 & 1.205\\
 \hline
 JNJ & 0.572 & 0.777 & 0.566 & 0.773 & -0.110 & 0.519 & -0.071 & 0.565 & 0.491 & 0.728 \\
 \hline
 JPM & 1.273 & 1.401 & 1.261 & 1.400 & -0.910 & 1.190 & -0.893 & 1.257 & 1.436 & 1.243 \\
 \hline
 KO & 0.237 & 1.016 & 0.224 & 1.023 & -0.227 & 0.356 & -0.108 & 0.400 & 0.411 & 0.657 \\
 \hline
 MCD & 0.686 & 0.745 & 0.669 & 0.736 & -0.903 & -0.478 & -0.787 & -0.366 & 0.674 & 0.795 \\
 \hline
 MMM & 0.778 & 1.098 & 0.772 & 1.103 & 0.585 & 1.145 & 0.588 & 1.154 & 0.881 & 1.016 \\
 \hline
 MRK & 0.850 & 0.877 & 0.842 & 0.874 & 0.727 & 0.359 & 0.854 & 0.394 & 0.705 & 0.725 \\
 \hline
 MSFT & 0.801 & 0.954 & 0.796 & 0.923 & -0.047 & -0.197 & -0.014 & -0.090 & 0.877 & 1.177 \\
 \hline
 NKE & 1.060 & 1.110 & 1.044 & 1.092 & -0.788 & -0.378 & -0.660 & -0.269 & 1.023 & 0.941 \\
 \hline
 PG & 0.642 & 0.467 & 0.637 & 0.456 & 0.152 & 0.107 & 0.201 & 0.159 & 0.576 & 0.634 \\
 \hline
 TRV & 1.363 & 1.397 & 1.355 & 1.407 & -0.511 & 0.551 & -0.415 & 0.593 & 0.947 & 0.926 \\
 \hline
 UNH & 0.880 & 1.203 & 0.877 & 1.182 & 0.685 & 0.180 & 0.842 & 0.245 & 0.683 & 0.925 \\
 \hline
 VZ & 0.757 & 0.467 & 0.745 & 0.464 & 0.414 & 0.285 & 0.542 & 0.289 & 0.741 & 0.625 \\
 \hline
 WBA & 0.666 & 1.208 & 0.664 & 1.225 & 0.223 & 0.856 & 0.320 & 0.848 & 0.644 & 0.820 \\
 \hline
 WMT & 0.688 & 0.198 & 0.681 & 0.180 & -0.916 & 1.007 & -0.884 & 0.976 & 0.739 & 0.550 \\
 \hline
\end{tabular}}
\label{tab1}
\end{table}

The results are shown in Table \ref{tab1} for $\rho_{H^2,0.2}$-Beta along with three other Beta metrics: $ERoD_{0+}$ Beta, $CDaR_{0.5}$ Beta, and the Standard Beta. We observe that the new Hellinger-Beta with confidence radius $\delta=0.2$ gives relatively similar results to Standard Beta since they are both derived based on returns, while the behavior is quite different from the two Drawdown Betas in both periods. The Hellinger-Beta differs from Standard Beta in values due to the different perspectives of risk they focus on. Specifically, the values reported for Hellinger-Betas in Table \ref{tab1} can be indicate more or less correlations with the market for different stocks and different periods. We also included the f-Betas calculated using deviations as a comparison. In addition, we report the Pearson correlation and Spearman correlation between the two columns of $\rho^{dev}_{H^2,0.2}$ for Period 1 and Period 2 below:
\begin{center}
Pearson Correlation: 0.318

Spearman Correlation: 0.545
\end{center}
The relatively strong correlations suggest that the Hellinger-Beta metrics are stable across different time periods.

The second set of tests we performed is to understand the sensitivity of the calculated Beta metric as the risk aversion parameter $\delta$ changes from small to relatively medium value, where the deviation-based Beta is used. Here we take the data in a single year as a test range and calculated $\rho^{dev}_{H^2,\delta}$ for different values of $\delta = 0.05,0.1,0.15,0.2,0.25,0.3,0.35$. Table \ref{tab2} summarizes the results of this experiment for the period $2006/01/03 - 2006/12/29$. Notice that for this period, most of the Hellinger-Beta values have a monotonic relationship as the risk aversion parameter is increased from 0.05 to 0.35. The monotonicity relationship can be positive or negative depending on the specific direction of movement of the particular asset with the market. The direction of change of the Beta values, as risk aversion parameter $\delta$ increases, can be seen as a measure of correlation between the asset and the market under increasingly more stressful scenarios. Assets with negative drift in the Beta values as stress level increases can be seen as good tools for hedging against large market downside risks, examples include CRM, GS, MCD, and MRK.

\begin{table}[!ht]

    \centering
      \caption{Different Risk Level Hellinger-Betas for DOW 30 Stocks in 2006}
      \scalebox{0.9}{
       \begin{tabular}{|| c | c | c | c | c | c | c | c ||} 
 \hline
  & $\rho^{dev}_{H^2,0.05}$ & $\rho^{dev}_{H^2,0.1}$ & $\rho^{dev}_{H^2,0.15}$ & $\rho^{dev}_{H^2,0.2}$ & $\rho^{dev}_{H^2,0.25}$ & $\rho^{dev}_{H^2,0.3}$ & $\rho^{dev}_{H^2,0.35}$ \\  
\hline
AAPL & 1.5613 & 1.5547 & 1.5679 & 1.5916 & 1.6214 & 1.6547 & 1.6902 \\
\hline
AMGN & 0.9676 & 1.1055 & 1.2079 & 1.2888 & 1.3562 & 1.4144 & 1.4661 \\
\hline
AXP & 1.0704 & 1.0936 & 1.113 & 1.1295 & 1.1438 & 1.1564 & 1.1677 \\
\hline
BA & 1.2913 & 1.3134 & 1.3231 & 1.3276 & 1.3298 & 1.331 & 1.3319 \\
\hline
CAT & 1.6078 & 1.6591 & 1.6947 & 1.7177 & 1.731 & 1.7364 & 1.735 \\
\hline
CRM & 1.6941 & 1.5299 & 1.4384 & 1.3889 & 1.3663 & 1.3624 & 1.3723 \\
\hline
CSCO & 1.3424 & 1.3565 & 1.3716 & 1.3874 & 1.4036 & 1.4196 & 1.4352 \\
\hline
CVX & 0.906 & 0.9221 & 0.9186 & 0.9039 & 0.8822 & 0.8558 & 0.8259 \\
\hline
DIS & 0.8469 & 0.8231 & 0.8067 & 0.7987 & 0.7983 & 0.8043 & 0.8163 \\
\hline
GS & 1.635 & 1.5502 & 1.4762 & 1.4127 & 1.3572 & 1.3076 & 1.2625 \\
\hline
HD & 0.9981 & 1.0113 & 1.0364 & 1.0641 & 1.0924 & 1.121 & 1.1498 \\
\hline
HON & 1.2996 & 1.3738 & 1.4405 & 1.4999 & 1.5533 & 1.6022 & 1.6475 \\
\hline
IBM & 0.8231 & 0.8104 & 0.8116 & 0.8207 & 0.8347 & 0.8524 & 0.873 \\
\hline
INTC & 1.3659 & 1.3066 & 1.2769 & 1.263 & 1.2588 & 1.2615 & 1.2694 \\
\hline
JNJ & 0.5138 & 0.5874 & 0.6506 & 0.706 & 0.7561 & 0.8025 & 0.8461 \\
\hline
JPM & 1.272 & 1.2805 & 1.289 & 1.2971 & 1.3051 & 1.3133 & 1.3217 \\
\hline
KO & 0.7207 & 0.7607 & 0.793 & 0.8217 & 0.8484 & 0.8741 & 0.8993 \\
\hline
MCD & 0.6063 & 0.4171 & 0.2655 & 0.1385 & 0.0272 & -0.0736 & -0.1674 \\
\hline
MMM & 0.9147 & 0.9656 & 1.0053 & 1.0366 & 1.0616 & 1.0818 & 1.0983 \\
\hline
MRK & 0.8299 & 0.8114 & 0.7754 & 0.7317 & 0.6849 & 0.6373 & 0.5896 \\
\hline
MSFT & 0.8842 & 0.9199 & 0.9448 & 0.9653 & 0.9842 & 1.0028 & 1.0216 \\
\hline
NKE & 0.5297 & 0.5469 & 0.5557 & 0.5583 & 0.5562 & 0.5504 & 0.5414 \\
\hline
PG & 0.6709 & 0.6892 & 0.6968 & 0.6993 & 0.699 & 0.6969 & 0.6935 \\
\hline
TRV & 0.9607 & 0.983 & 1.0117 & 1.0386 & 1.0623 & 1.0823 & 1.0989 \\
\hline
UNH & 0.739 & 0.8158 & 0.8942 & 0.97 & 1.0424 & 1.1116 & 1.178 \\
\hline
VZ & 0.9604 & 0.9674 & 0.9648 & 0.9573 & 0.9477 & 0.9374 & 0.9272 \\
\hline
WBA & 0.9219 & 0.9876 & 1.0417 & 1.0892 & 1.1326 & 1.1732 & 1.2118 \\
\hline
WMT & 0.9784 & 0.9909 & 0.988 & 0.9799 & 0.9704 & 0.9611 & 0.9527 \\
\hline
\end{tabular}}
\label{tab2}
\end{table}
\begin{figure}[!ht]
    \centering
    \includegraphics[scale=0.5]{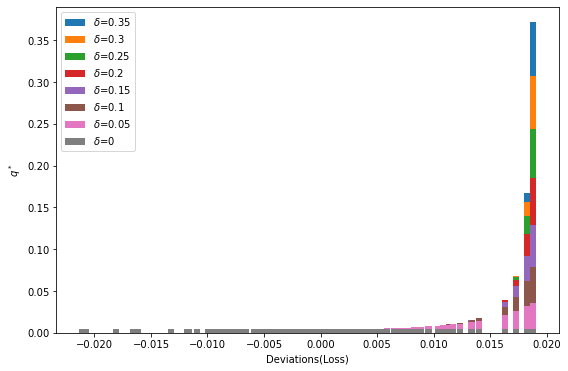}
    \caption{$\rho^{dev}_{H^2,\delta}$ Optimal Perturbed Market Probabilities for Varying $\delta$ in 2006}
    \label{fig:f_beta_q}
\end{figure}

We also observe that for this one-year period, the risk aversion parameter $\delta = 0.35$ is actually large enough in the sense that the optimal perturbation measure allocates more than 50\% of the probabilities onto the top two deviation loss scenarios, which is smaller than 1\% of the total number of dates. The perturbation allocation for different risk aversion parameters in the deviation-based Hellinger Beta calculation can be visualized in Figure \ref{fig:f_beta_q} based on the loss observations (x-axis) in the S\&P 500 market index from the year 2006. The y-axis plots the optimal risk identifiers in bars for each loss observation and for different risk aversion parameters $\delta$, shown in different colors. We can see that increasing levels of $\delta$ will perturb the probability distribution of $q^\star$ towards more extreme deviation (loss) scenarios. Hence we conclude that, in general, a risk aversion level of $\delta = 0.3$ is large enough for enough stress to be put on extreme worst cases of the deviation (loss) scenarios. Keep increasing the risk aversion parameter won't mean much as we just keep adding probabilities onto the top few extreme scenarios. Instead, a wider range of behavior is observable by considering smaller values of the risk aversion parameter as the perturbed probabilities shift gradually from the uniform distribution towards a more concentrated one.

\begin{table}[!ht]

    \centering
      \caption{Different Risk Level Hellinger-Betas for DOW 30 Stocks in 2013}
      \scalebox{0.9}{
       \begin{tabular}{|| c | c | c | c | c | c | c | c ||} 
 \hline
  & $\rho^{dev}_{H^2,0.05}$ & $\rho^{dev}_{H^2,0.1}$ & $\rho^{dev}_{H^2,0.15}$ & $\rho^{dev}_{H^2,0.2}$ & $\rho^{dev}_{H^2,0.25}$ & $\rho^{dev}_{H^2,0.3}$ & $\rho^{dev}_{H^2,0.35}$ \\ 
\hline
AAPL & 0.7453 & 0.7348 & 0.7134 & 0.6929 & 0.6747 & 0.6588 & 0.645 \\
\hline
AMGN & 1.0872 & 1.1012 & 1.1158 & 1.1276 & 1.1367 & 1.144 & 1.1498 \\
\hline
AXP & 0.9862 & 0.9287 & 0.8918 & 0.8652 & 0.8445 & 0.8279 & 0.8141 \\
\hline
BA & 1.0219 & 1.044 & 1.0553 & 1.062 & 1.0663 & 1.0693 & 1.0714 \\
\hline
CAT & 0.8569 & 0.8017 & 0.7529 & 0.7114 & 0.6759 & 0.6453 & 0.6187 \\
\hline
CRM & 1.2519 & 1.1803 & 1.1313 & 1.0964 & 1.0702 & 1.0499 & 1.0336 \\
\hline
CSCO & 0.8064 & 0.722 & 0.6593 & 0.6116 & 0.5737 & 0.5429 & 0.5172 \\
\hline
CVX & 0.9013 & 0.9094 & 0.9068 & 0.901 & 0.8943 & 0.8876 & 0.8811 \\
\hline
DIS & 1.2187 & 1.2778 & 1.3191 & 1.3493 & 1.3726 & 1.3912 & 1.4063 \\
\hline
GS & 1.4248 & 1.4093 & 1.4099 & 1.417 & 1.4265 & 1.4368 & 1.4471 \\
\hline
HD & 0.9557 & 0.9862 & 1.0059 & 1.0205 & 1.0322 & 1.0418 & 1.05 \\
\hline
HON & 1.1175 & 1.0991 & 1.0868 & 1.0774 & 1.0698 & 1.0633 & 1.0578 \\
\hline
IBM & 0.7449 & 0.7532 & 0.7664 & 0.7799 & 0.7926 & 0.8042 & 0.8147 \\
\hline
INTC & 0.9239 & 0.9882 & 1.0397 & 1.0811 & 1.115 & 1.1435 & 1.1678 \\
\hline
JNJ & 0.7964 & 0.8415 & 0.8791 & 0.9095 & 0.9346 & 0.9557 & 0.9737 \\
\hline
JPM & 1.1035 & 1.0251 & 0.9752 & 0.9407 & 0.9153 & 0.8956 & 0.88 \\
\hline
KO & 0.9118 & 0.9982 & 1.0549 & 1.0946 & 1.1242 & 1.147 & 1.1653 \\
\hline
MCD & 0.586 & 0.6317 & 0.6667 & 0.6932 & 0.714 & 0.7307 & 0.7444 \\
\hline
MMM & 0.9247 & 0.9425 & 0.9568 & 0.9684 & 0.978 & 0.9861 & 0.9931 \\
\hline
MRK & 0.7537 & 0.8198 & 0.8673 & 0.9033 & 0.9318 & 0.955 & 0.9744 \\
\hline
MSFT & 0.888 & 0.935 & 0.9826 & 1.0253 & 1.0629 & 1.0959 & 1.125 \\
\hline
NKE & 0.7917 & 0.7521 & 0.7303 & 0.7165 & 0.7068 & 0.6997 & 0.6942 \\
\hline
PG & 0.8004 & 0.8593 & 0.9118 & 0.9564 & 0.9945 & 1.0272 & 1.0556 \\
\hline
TRV & 1.0165 & 1.0479 & 1.074 & 1.0955 & 1.1135 & 1.1288 & 1.142 \\
\hline
UNH & 0.9071 & 0.922 & 0.9268 & 0.9281 & 0.9281 & 0.9274 & 0.9264 \\
\hline
VZ & 0.6337 & 0.6496 & 0.6757 & 0.7015 & 0.7248 & 0.7457 & 0.7642 \\
\hline
WBA & 1.0009 & 1.0795 & 1.1458 & 1.202 & 1.2501 & 1.2917 & 1.3281 \\
\hline
WMT & 0.4818 & 0.5169 & 0.5526 & 0.5842 & 0.6116 & 0.6354 & 0.6561 \\
\hline
\end{tabular}}
\label{tab3}
\end{table}
 \begin{figure}[ht!]
\centering
\begin{subfigure}{.5\textwidth}
  \centering
  \includegraphics[width=1\textwidth]{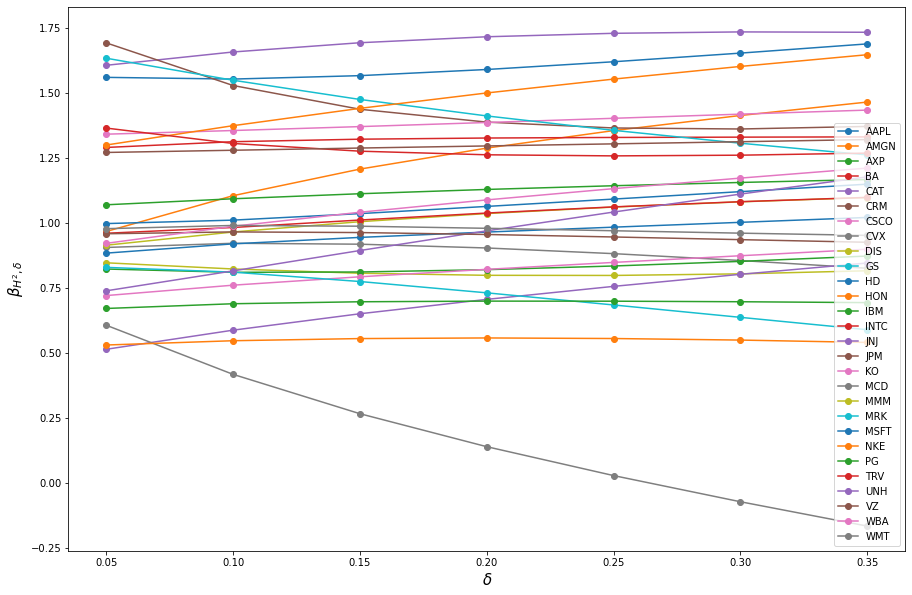}
    \caption{2006}
\end{subfigure}%
\begin{subfigure}{.5\textwidth}
  \centering
  \includegraphics[width=1\textwidth]{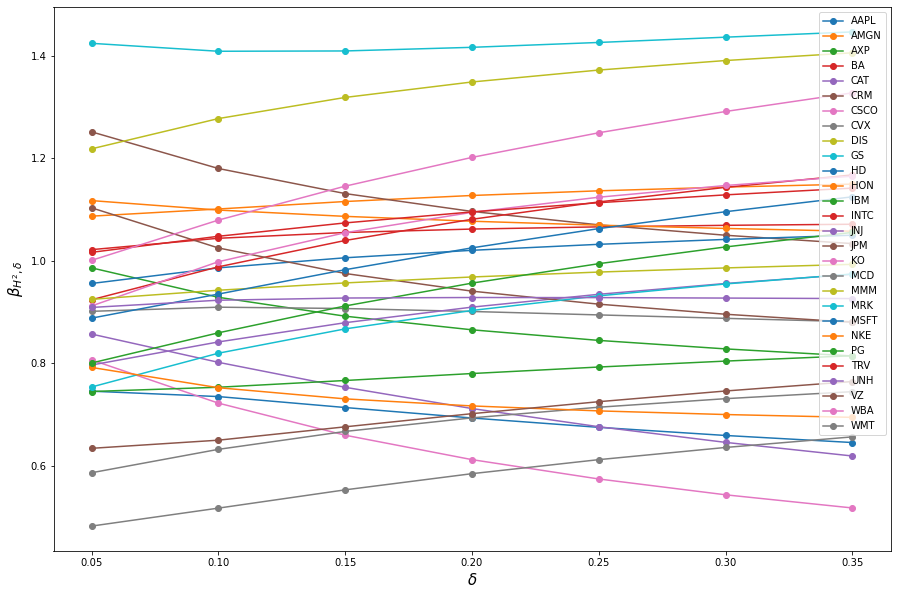}
    \caption{2013}
\end{subfigure}%
\caption{$\beta_{H^2,\delta}$ Drift against $\delta$ for DOW 30 Stocks}
\label{fig:HBeta_move}
\end{figure}

The same experiment as in Table \ref{tab2} is carried out for the period $2013/01/03 - 2013/12/31$. We arrive at very similar observations as the movement of Hellinger-Beta values (based on deviations) are mostly monotonic as the risk aversion parameter is increased. The results are summarized in Table \ref{tab3}.

We also make the observation that while most reported Hellinger-Beta values are monotonic with respect to changing risk aversion parameters, the direction of movement can be different in 2013 from what is reported in 2006. Hence this indicates that meaningful hedging can be obtained by looking at the direction of Hellinger-Beta drifts, which adapts to stress scenarios under different market conditions in different time periods. Figure \ref{fig:HBeta_move}(a) (and \ref{fig:HBeta_move}(b)) demonstrates the general monotonic drift behavior of (deviation-based) Hellinger-Beta values for changing $\delta$, reported respectively in Table \ref{tab2} (and Table \ref{tab3}) for each stock in the year 2006 (and the year 2013). The risk aversion parameter $\delta$ is plotted on the x-axis in each panel, and the Hellinger-Beta values for each DOW 30 stock are plotted on the y-axis.

\newpage
\section{Extension to f-Drawdown Betas and Numerical Results}

Similar to \citet{DU22}, we can extend the f-Betas for a single path to a setting where the drawdown random variables are considered and the f-Beta would translate to f-Drawdown Beta defined as follows:
\begin{equation}
\label{fbeta_drawdown}
\beta_{\phi,\delta}^{i,DD} = \frac{\sum_{t=1}^{T}q_{t}^\star dd_{t}^i}{\rho_{\phi,\delta}(dd^M)}\;,
\end{equation}
where the drawdown variables for the market index are defined as, 
\begin{equation}\label{dd_market}
dd_t^M = \max_{1\leq\tau\leq t}w^M_\tau - w^M_t
\end{equation}
for each time step $t$, with uncompounded cumulative returns of the market at time $t$ defined as,
$$w^M_t = \sum_{\tau=1}^t r^M_\tau$$  Let $\tau(t)$ denote the historic peak location in the previous definition for drawdowns, that is to say, $w^M_{\tau(t)} = \max_{1\leq\tau\leq t}w^M_\tau$. Then the change in cumulative returns in the market drawdown period defined at any given time $t$ for an individual asset $i$ is defined as,
\begin{equation}\label{dd_asset}
dd_t^i = -\sum_{\tau(t)\leq\tau\leq t}r_\tau^i
\end{equation}

The rest of the terms in \eqref{fbeta_drawdown} are defined similarly to those from previous f-Betas such as in \eqref{fbeta_3}. The proof for optimality conditions of drawdown type Betas is provided in \citet{DU22}, \citet{13}, and the optimality condition of f-divergence induced risk measures on drawdown observations can be obtained by considering the prior results in conjunction with general theories about deviation based Betas, see \citet{14}. While previous works on CDaR-Beta (\citet{13}) and ERoD-Beta (\citet{DU22}) are based on risk measures (applied to drawdown deviations) such as CVaR, our new approach uses the f-divergence induced risk measure $\rho_{\phi,\delta}$ which provides nicer properties and an intuitive interpretation from the point of view of distributional robustness.

Similar to Section 6, we performed several numerical experiments on f-Drawdown Betas, using the S\&P 500 index as the optimal market portfolio. Only the results for DOW 30 stocks are presented for demonstration. In particular, the generating f-divergence is chosen to be the squared Hellinger distance, which gives the Hellinger-Drawdown Beta:

\begin{equation}
\label{hbeta_drawdown}
\beta_{H^2,\delta}^{i,DD} = \frac{\sum_{t=1}^{T}q_{t}^\star dd_{t}^i}{\rho_{H^2,\delta}(dd^M)}\;,
\end{equation}

In the first experiment, consider the same two non-overlapping periods, Period 1 and Period 2, from Section 6. We reproduced the part of Table \ref{tab1} corresponding to Drawdown Betas below in Table \ref{tab4}, but now introducing the column labeled $\rho^{DD}_{H^2,0.2}$ that stands for Hellinger-Drawdown Beta calculated with $\delta=0.2$. The results are shown in Table \ref{tab4} for $\rho^{DD}_{H^2,0.2}$-Beta along with two other Drawdown Beta metrics: $ERoD_{0+}$ Beta and $CDaR_{0.5}$ Beta. It can be seen that for some assets, the Hellinger risk measure offers different behaviors than CVaR-type risk measures on drawdown variables, but they also share some similarities on various assets. In addition, we report the Pearson correlation and Spearman correlation between the two columns of $\rho^{DD}_{H^2,0.2}$ for Period 1 and Period 2 below:
\begin{center}
Pearson Correlation: 0.469

Spearman Correlation: 0.284

\end{center}

\begin{table}[!ht]

    \centering
      \caption{Drawdown Betas for DOW 30 Stocks: Non-overlapping Period 1 and 2 }
      \scalebox{0.9}{
       \begin{tabular}{|| c | c | c | c | c | c | c ||} 
 \hline
 & $\rho^{DD}_{H^2,0.2}$ & $\rho^{DD}_{H^2,0.2}$ & $ERoD_{0+}$ & $ERoD_{0+}$ & $CDaR_{0.5}$
 & $CDaR_{0.5}$ \\  

  &  Period 1 & Period 2 &  Period 1 & Period 2 & Period 1  & Period 2 \\
 \hline
  AAPL & 0.316 & 0.834 & -0.986 & 1.126  & -0.796 & 1.207 \\ 
 \hline
AMGN & -0.022 & 0.364 & -0.218 & 0.468  & -0.163 & 0.523 \\
 \hline
 AXP & 1.574 & 1.444 & 0.540 & 1.606  & 0.659 & 1.547 \\
 \hline
 BA & 1.429 & 2.797 & 1.332 & 1.277 & 1.376 & 1.342  \\
 \hline
 CAT & 1.243 & 0.975 & 0.558 & 1.892 & 0.626 & 1.868 \\
 \hline
 CRM & 0.132 & 0.714 & -1.438 & 0.154 & -1.268 & 0.207 \\
 \hline
 CSCO & 0.984 & 0.596 & 1.119 & 0.664 & 1.015 & 0.641  \\
 \hline
 CVX & 0.210 & 1.417 & -0.129 & 1.323 & -0.038 & 1.303 \\
 \hline
 DIS & 0.673 & 1.142 & 0.081 & 0.859 & 0.197 & 0.893  \\
 \hline
 GS & 0.867 & 1.315 & 0.640 & 1.956 & 0.564 & 1.896  \\
 \hline
 HD & 0.384 & 0.957 & -0.037 & 0.253 & 0.040 & 0.373  \\
 \hline
 HON & 0.993 & 1.100 & 0.768 & 0.967 & 0.830 & 1.015  \\
 \hline
 IBM & 0.146 & 1.212 & -0.410 & 1.986 & -0.337 & 1.901 \\
 \hline
 INTC & 0.639 & 0.763 & 0.317 & 0.405 & 0.323 & 0.386 \\
 \hline
 JNJ & 0.189 & 0.481 & -0.110 & 0.519 & -0.071 & 0.565 \\
 \hline
 JPM & 0.125 & 1.164 & -0.910 & 1.190 & -0.893 & 1.257 \\
 \hline
 KO & 0.265 & 0.962 & -0.227 & 0.356 & -0.108 & 0.400  \\
 \hline
 MCD & -0.240 & 0.706 & -0.903 & -0.478 & -0.787 & -0.366 \\
 \hline
 MMM & 0.829 & 0.616 & 0.585 & 1.145 & 0.588 & 1.154 \\
 \hline
 MRK & 0.877 & 0.302 & 0.727 & 0.359 & 0.854 & 0.394  \\
 \hline
 MSFT & 0.487 & 0.585 & -0.047 & -0.197 & -0.014 & -0.090  \\
 \hline
 NKE & 0.096 & 0.858 & -0.788 & -0.378 & -0.660 & -0.269  \\
 \hline
 PG & 0.378 & 0.384 & 0.152 & 0.107 & 0.201 & 0.159 \\
 \hline
 TRV & 0.016 & 0.858 & -0.511 & 0.551 & -0.415 & 0.593 \\
 \hline
 UNH & 0.956 & 0.723 & 0.685 & 0.180 & 0.842 & 0.245  \\
 \hline
 VZ & 0.451 & 0.260 & 0.414 & 0.285 & 0.542 & 0.289  \\
 \hline
 WBA & 0.526 & 0.284 & 0.223 & 0.856 & 0.320 & 0.848 \\
 \hline
 WMT & -0.409 & 0.106 & -0.916 & 1.007 & -0.884 & 0.976 \\
 \hline
\end{tabular}}
\label{tab4}
\end{table}

Similar to Table \ref{tab3}, we performed the risk level sensitivity test for Hellinger-Drawdown Beta in the year 2013 and the year 2022, which have quite different market drawdown behaviors in the S\&P 500 index. Specifically, the year 2022 suffers from serious global financial instability as global markets were also impacted by fears of economic recession, resulting in large drawdowns in the S\&P 500 index throughout the year. Similar to the observations we made with (deviation-based) Hellinger-Betas in Section 6, the Hellinger-Drawdown Betas for increasing risk aversion level tend to produce a monotonic shift relationship, which is indicative of the direction of correlation of the asset and the market in more and more extreme drawdown periods. Figure \ref{fig:HDBeta_move} demonstrates the general monotonic drift behavior of Hellinger-Drawdown Betas calculated with changing $\delta$ values for each of the DOW 30 stocks in the year 2013, as reported in Table \ref{tab5}, and the year 2022, as reported in Table \ref{tab6}. The risk aversion parameter $\delta$ is plotted on the x-axis in each panel, and the Hellinger-Beta values for each DOW 30 stock are plotted on the y-axis. The behavior of the values and drifts of Hellinger-Drawdown Betas is different in these two years, and this should be expected given the different market conditions in these two years. For instance, for the risk aversion parameters considered, the changes in Hellinger-Drawdown Betas in 2022 are less significant due to the market being in serious drawdown periods throughout the year.

\begin{figure}[ht!]
\centering
\begin{subfigure}{.5\textwidth}
  \centering
  \includegraphics[width=1\textwidth]{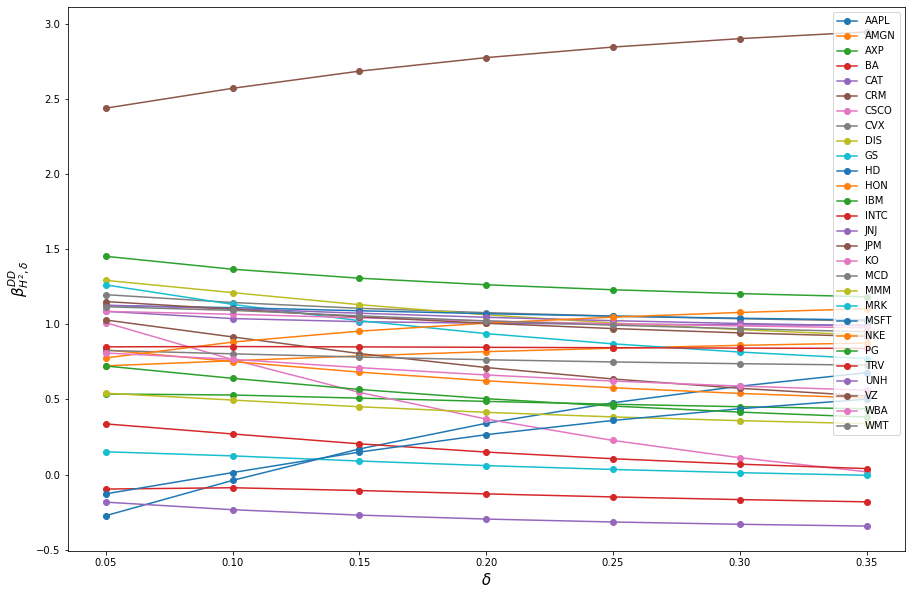}
    \caption{2013}
\end{subfigure}%
\begin{subfigure}{.5\textwidth}
  \centering
  \includegraphics[width=1\textwidth]{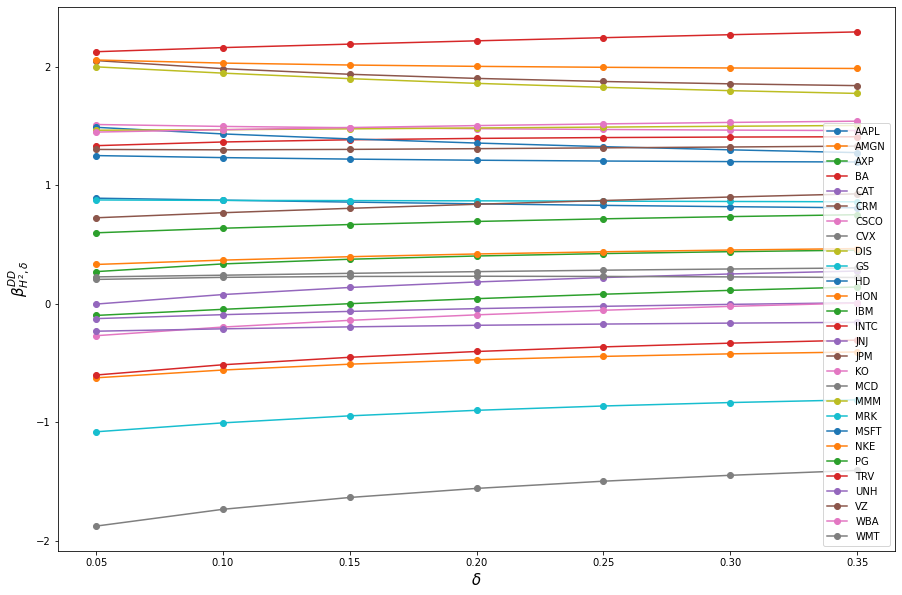}
    \caption{2022}
\end{subfigure}%
\caption{$\beta^{DD}_{H^2,\delta}$ Drift against $\delta$ for DOW 30 Stocks in 2013}
\label{fig:HDBeta_move}
\end{figure}

\begin{table}[!ht]

    \centering
      \caption{Different Risk Level Hellinger-Drawdown Betas for DOW 30 Stocks in 2013}
      \scalebox{0.9}{
       \begin{tabular}{|| c | c | c | c | c | c | c | c ||} 
 \hline
  & $\rho^{DD}_{H^2,0.05}$ & $\rho^{DD}_{H^2,0.1}$ & $\rho^{DD}_{H^2,0.15}$ & $\rho^{DD}_{H^2,0.2}$ & $\rho^{DD}_{H^2,0.25}$ & $\rho^{DD}_{H^2,0.3}$ & $\rho^{DD}_{H^2,0.35}$ \\ 
\hline
AAPL & -0.2736 & -0.0389 & 0.1706 & 0.3412 & 0.4779 & 0.5882 & 0.6781 \\
\hline
AMGN & 0.7235 & 0.7576 & 0.7901 & 0.818 & 0.8411 & 0.86 & 0.8756 \\
\hline
AXP & 0.5362 & 0.5296 & 0.5086 & 0.487 & 0.4679 & 0.4515 & 0.4378 \\
\hline
BA & -0.0965 & -0.0877 & -0.1064 & -0.1286 & -0.149 & -0.1666 & -0.1816 \\
\hline
CAT & 1.0864 & 1.0382 & 1.0177 & 1.0071 & 1.0008 & 0.9967 & 0.9939 \\
\hline
CRM & 2.4388 & 2.571 & 2.6852 & 2.7751 & 2.8455 & 2.9014 & 2.9463 \\
\hline
CSCO & 1.0111 & 0.7651 & 0.547 & 0.3693 & 0.2269 & 0.1119 & 0.0182 \\
\hline
CVX & 1.1974 & 1.1451 & 1.1067 & 1.0776 & 1.055 & 1.0371 & 1.0227 \\
\hline
DIS & 1.2926 & 1.2109 & 1.1307 & 1.0635 & 1.0088 & 0.9644 & 0.928 \\
\hline
GS & 1.2623 & 1.1326 & 1.0242 & 0.9381 & 0.87 & 0.8154 & 0.7713 \\
\hline
HD & 1.1244 & 1.1111 & 1.0903 & 1.0709 & 1.0543 & 1.0405 & 1.029 \\
\hline
HON & 0.8268 & 0.75 & 0.6801 & 0.6227 & 0.5765 & 0.5391 & 0.5086 \\
\hline
IBM & 1.4503 & 1.3647 & 1.3055 & 1.262 & 1.229 & 1.2033 & 1.1828 \\
\hline
INTC & 0.3372 & 0.2702 & 0.2046 & 0.1497 & 0.1052 & 0.0692 & 0.0397 \\
\hline
JNJ & 1.1276 & 1.1041 & 1.0739 & 1.0468 & 1.0242 & 1.0056 & 0.9902 \\
\hline
JPM & 1.0288 & 0.916 & 0.8055 & 0.7123 & 0.6362 & 0.5742 & 0.5234 \\
\hline
KO & 1.0862 & 1.067 & 1.0433 & 1.0221 & 1.0042 & 0.9894 & 0.9771 \\
\hline
MCD & 0.8262 & 0.8037 & 0.7821 & 0.7641 & 0.7497 & 0.7379 & 0.7283 \\
\hline
MMM & 0.5414 & 0.4959 & 0.4514 & 0.4141 & 0.3837 & 0.3589 & 0.3386 \\
\hline
MRK & 0.147 & 0.1191 & 0.0845 & 0.0534 & 0.0274 & 0.0059 & -0.0119 \\
\hline
MSFT & -0.1284 & 0.0132 & 0.1497 & 0.2655 & 0.3604 & 0.4381 & 0.502 \\
\hline
NKE & 0.7762 & 0.8818 & 0.9546 & 1.0077 & 1.0478 & 1.079 & 1.1038 \\
\hline
PG & 0.7226 & 0.6405 & 0.5664 & 0.5052 & 0.4558 & 0.4157 & 0.3829 \\
\hline
TRV & 0.8501 & 0.852 & 0.8498 & 0.847 & 0.8442 & 0.8419 & 0.8398 \\
\hline
UNH & -0.2243 & -0.281 & -0.3226 & -0.3534 & -0.3767 & -0.3948 & -0.4092 \\
\hline
VZ & 1.1515 & 1.1007 & 1.0499 & 1.0069 & 0.9717 & 0.943 & 0.9195 \\
\hline
WBA & 0.809 & 0.7669 & 0.7117 & 0.6628 & 0.6223 & 0.5891 & 0.5618 \\
\hline
WMT & 1.117 & 1.0928 & 1.0564 & 1.0228 & 0.9943 & 0.9706 & 0.951 \\
\hline
\end{tabular}}
\label{tab5}
\end{table}

\begin{table}[!ht]

    \centering
      \caption{Different Risk Level Hellinger-Drawdown Betas for DOW 30 Stocks in 2022}
      \scalebox{0.9}{
       \begin{tabular}{|| c | c | c | c | c | c | c | c ||} 
 \hline
  & $\rho^{DD}_{H^2,0.05}$ & $\rho^{DD}_{H^2,0.1}$ & $\rho^{DD}_{H^2,0.15}$ & $\rho^{DD}_{H^2,0.2}$ & $\rho^{DD}_{H^2,0.25}$ & $\rho^{DD}_{H^2,0.3}$ & $\rho^{DD}_{H^2,0.35}$ \\ 
  \hline
AAPL & 0.8916 & 0.8749 & 0.8586 & 0.8439 & 0.8309 & 0.8198 & 0.8104 \\
\hline
AMGN & -0.6252 & -0.5594 & -0.5101 & -0.4727 & -0.4444 & -0.4229 & -0.4068 \\
\hline
AXP & 0.2715 & 0.3366 & 0.3761 & 0.4034 & 0.4239 & 0.4399 & 0.4529 \\
\hline
BA & 1.3351 & 1.3661 & 1.3848 & 1.3964 & 1.4035 & 1.4078 & 1.4101 \\
\hline
CAT & -0.0025 & 0.0779 & 0.1377 & 0.1843 & 0.2215 & 0.2517 & 0.2765 \\
\hline
CRM & 2.0536 & 1.9853 & 1.9378 & 1.903 & 1.8769 & 1.8571 & 1.842 \\
\hline
CSCO & 1.5137 & 1.4975 & 1.4866 & 1.4783 & 1.4717 & 1.4665 & 1.4624 \\
\hline
CVX & -1.8773 & -1.7355 & -1.6356 & -1.5591 & -1.4982 & -1.4484 & -1.407 \\
\hline
DIS & 2.0005 & 1.9475 & 1.9009 & 1.8611 & 1.8276 & 1.7995 & 1.7762 \\
\hline
GS & 0.8766 & 0.8727 & 0.8712 & 0.8692 & 0.8666 & 0.8638 & 0.8612 \\
\hline
HD & 1.4899 & 1.4341 & 1.3922 & 1.3571 & 1.3266 & 1.2999 & 1.2764 \\
\hline
HON & 0.332 & 0.3685 & 0.3975 & 0.4205 & 0.4389 & 0.4538 & 0.4659 \\
\hline
IBM & -0.0984 & -0.047 & 0.0007 & 0.0432 & 0.0806 & 0.1133 & 0.1418 \\
\hline
INTC & 2.1284 & 2.1631 & 2.1925 & 2.2203 & 2.2468 & 2.272 & 2.2956 \\
\hline
JNJ & -0.1253 & -0.092 & -0.0643 & -0.041 & -0.0215 & -0.0051 & 0.0087 \\
\hline
JPM & 1.3031 & 1.2999 & 1.3036 & 1.3096 & 1.3166 & 1.3238 & 1.3313 \\
\hline
KO & -0.2696 & -0.1966 & -0.1399 & -0.0935 & -0.0547 & -0.0218 & 0.0062 \\
\hline
MCD & 0.227 & 0.2411 & 0.257 & 0.2713 & 0.2835 & 0.2937 & 0.3023 \\
\hline
MMM & 1.4654 & 1.4701 & 1.4769 & 1.4843 & 1.4918 & 1.4989 & 1.5056 \\
\hline
MRK & -1.0809 & -1.0056 & -0.9463 & -0.8997 & -0.8631 & -0.8341 & -0.8111 \\
\hline
MSFT & 1.2513 & 1.2341 & 1.2216 & 1.2122 & 1.2053 & 1.2005 & 1.1973 \\
\hline
NKE & 2.0585 & 2.0324 & 2.016 & 2.0047 & 1.9968 & 1.991 & 1.9868 \\
\hline
PG & 0.5992 & 0.6378 & 0.669 & 0.6951 & 0.7172 & 0.736 & 0.7521 \\
\hline
TRV & -0.6016 & -0.5149 & -0.4519 & -0.403 & -0.3641 & -0.3325 & -0.3067 \\
\hline
UNH & -0.2317 & -0.2117 & -0.195 & -0.1818 & -0.1714 & -0.1633 & -0.1568 \\
\hline
VZ & 0.7261 & 0.7689 & 0.8059 & 0.8402 & 0.872 & 0.9014 & 0.9283 \\
\hline
WBA & 1.4498 & 1.4697 & 1.4881 & 1.5046 & 1.5189 & 1.5313 & 1.5418 \\
\hline
WMT & 0.2057 & 0.223 & 0.2306 & 0.2324 & 0.2308 & 0.227 & 0.2219 \\
\hline
\end{tabular}}
\label{tab6}
\end{table}

\clearpage
\newpage

\section{Conclusions}

This paper builds on using the class of f-divergence induced coherent risk measures \cite{DP20} for portfolio optimization and derives its necessary optimality conditions formulated in the CAPM format. We derived a new f-Beta  similar to the Standard Beta, CVaR Beta, and previous works in CDaR Beta \cite{13} and ERoD Beta \cite{DU22}. The f-Beta evaluates portfolio performance under an optimally perturbed market probability measure, and this family of Beta metrics gives various degrees of flexibility and interpretability. In particular, the choice of divergence-generating function $\phi$ controls the shape of the risk function, which dictates the sensitivity of the proposed Beta metric to changes in the risk aversion degree parameter $\delta$. We conducted a numerical experiment using DOW 30 stocks returns against a chosen optimal market portfolio to demonstrate the new perspectives provided by Hellinger-Beta as compared with Standard Beta and Drawdown Betas, based on choosing the squared Hellinger distance to be the particular choice of f-divergence function in the general f-divergence induced risk measures and f-Betas. Our results show that the Hellinger-Betas (calculated based on deviations) provide different perspectives on the market from previous metrics such as CDaR Beta, ERoD Beta, and Standard Beta. For various choices of risk aversion parameter $\delta$, the Beta values generally have a monotonic shift behavior which may indicate the positive or negative relation it has with the market when the market moves into more stressful scenarios. Such a relation can be useful for hedging purposes. Similar to risk measures on drawdown observations and related Beta metrics (\citet{2}, \citet{13}, \citet{DU22}), we can develop f-divergence induced drawdown risk measures and Drawdown Betas based on a chosen risk measure from a generating f-divergence function. This can be seen as a similar approach to CDaR Beta \cite{13} and ERoD Beta \cite{DU22} as proposed in earlier works, but now with the family of f-divergence induced risk measures providing more flexibility and richer risk aversion behaviors than the CVaR risk measures. Similarly, we provided numerical results for Hellinger-Drawdown Beta by choosing the squared Hellinger distance to be the generating f-divergence.

Future works include studying the impact of the choice of f-divergences on the resulting Beta metric. Based on their statistical properties, the squared Hellinger distance is chosen in this study to demonstrate the behavior of the family of f-divergences since they provide bounds for meaningful statistical estimation and have various desirable properties, such as being symmetric and bounded. However, a comparison with other choices of f-divergence functions (and their induced risk measures for the CAPM setting) will reveal further insights into how the shape of the risk function impacts the sensitivity of the resulting f-Beta metrics. In particular, we can study if there is an overarching optimization framework under which the choice of $\alpha$ in the family of $\alpha$-divergences can be jointly optimized under some criteria. The relation between f-divergence induced risk measures and CVaR risks from the spectral risk measure point of view also suggests that we can compare and understand f-Betas and CVaR Betas in a more systematic way. In recent work, \citet{FW22} proposed a general framework to construct specific divergence-induced risk measures that exhibit a desired tail sensitivity. Consequently, this result can be potentially used in our framework to obtain f-Betas that have a more explainable and controlled risk aversion behavior. Another interesting future work would be to quantify the incremental changes in the f-Beta values as the risk aversion parameter increases and use the generated marginal increment/decrement as a measure of directional movement between the asset and the market. This can be a more useful metric for hedging purposes under higher-stress situations. Finally, we propose to study further the behaviors of f-Drawdown Betas as compared with other drawdown-based Betas and quantify their incremental changes with respect to risk aversion parameter changes as well.

\section*{Data Availability Statement}
Data for the case study were downloaded from Yahoo Finance \url{https://finance.yahoo.com/}. 

\section*{Acknowledgements}
The author thanks Dr. Andrew Mullhaupt (Stony Brook University), Dr. Stan Uryasev (Stony Brook University), Dr. Michael Albert (University of Virginia), and anonymous reviewers for valuable discussions and comments.

\end{document}